\documentclass[final,onefignum,onetabnum]{siamonline190516}

\usepackage{lipsum}
\usepackage{amsfonts}
\usepackage{graphicx}
\usepackage{epstopdf}
\usepackage{algorithmic}
\ifpdf
  \DeclareGraphicsExtensions{.eps,.pdf,.png,.jpg}
\else
  \DeclareGraphicsExtensions{.eps}
\fi

\usepackage{enumitem}
\setlist[enumerate]{leftmargin=.5in}
\setlist[itemize]{leftmargin=.5in}

\newsiamremark{remark}{Remark}
\newsiamremark{hypothesis}{Hypothesis}
\crefname{hypothesis}{Hypothesis}{Hypotheses}
\newsiamthm{claim}{Claim}

\headers{Modularity maximisation for graphons}{F. Klimm, Nick S. Jones, and Michael T. Schaub}

\title{Modularity maximisation for  graphons\thanks{Submitted to the editors \today.
\funding{F.K. and N.S.J. thank the EPSRC (Centre for Mathematics of Precision Healthcare; EP/N014529/1). 
M.T.S. acknowledges funding from the Ministry of Culture and Science (MKW) of the German State of North Rhine-Westphalia ("NRW Rückkehrprogramm").}} 
}

\author{Florian Klimm\thanks{Department of Mathematics, Imperial College London \& MRC Mitochondrial Biology Unit, University of Cambridge
  (\email{f.klimm@imperial.ac.uk}).}
\and Nick S. Jones\thanks{Department of Mathematics, Imperial College London 
  (\email{nick.jones@imperial.ac.uk}).}
\and Michael T. Schaub\thanks{Department of Computer Science, RWTH Aachen University
  (\email{schaub@cs.rwth-aachen.de}).}}

\usepackage{amsopn}

\usepackage[showframe=false]{geometry}

\usepackage{amsmath}
\usepackage{amssymb}
\usepackage{xcolor}
\usepackage{graphicx}
\usepackage{rotating}

\usepackage{pdflscape}
\usepackage{afterpage}
\usepackage{capt-of}
\usepackage{url}

\usepackage{booktabs}
\usepackage{array}
\usepackage{threeparttable}

\begin{document}
\maketitle

\begin{abstract} 
Networks are a widely-used tool to investigate the large-scale connectivity structure in complex systems and 
\emph{graphons} have been proposed as an infinite size limit of dense networks. The detection of communities or other meso-scale structures is a prominent topic in network science as it allows the identification of functional building blocks in complex systems. When such building blocks may be present in graphons is an open question. In this paper, we define a graphon-modularity and demonstrate that it can be maximised to detect communities in graphons. We then investigate specific synthetic graphons and show that they may show a wide range of different community structures. We also reformulate the graphon-modularity maximisation as a continuous optimisation problem and so prove the optimal community structure or lack thereof for some graphons, something that is usually not possible for networks. Furthermore, we demonstrate that estimating a graphon from network data as an intermediate step can improve the detection of communities, in comparison with exclusively maximising the modularity of the network. While the choice of graphon-estimator may strongly influence the accord between the community structure of a network and its estimated graphon, we find that there is a substantial overlap if an appropriate estimator is used. Our study demonstrates that community detection for graphons is possible and may serve as a privacy-preserving way to cluster network data.
 \end{abstract}



\section{Introduction}

Networks have become nearly ubiquitous abstractions for complex systems arising in biological, technical, and social contexts, and many other applications~\cite{newman2010networks}.
Mathematically, such networks are represented as graphs in which nodes represent entities and edges relations between such entities.
Accordingly, graph-based tools have been employed to study and reveal properties of such network systems.
Of particular interest has been the detection of \emph{community structure}, i.e., a grouping of nodes that are more similar to each other according to some criterion than to the rest of the network.
There is is no standard definition of community structure~\cite{fortunato2010community,fortunato2016community,schaub2017many} and different notions of community structure exist. We adopt here the common viewpoint of defining communities as densely connected groups of nodes, when compared to the remainder of the network.

Detecting such communities has led to insights about the function of proteins~\cite{lewis2010function}, social networks~\cite{traud2012social}, neuroscience~\cite{bassett2017network}, and many other fields~\cite{porter2009communities}. 
A wide range of community-detection algorithms exists, such as maximum-likelihood estimation of generative network models, spectral methods, and matrix-decomposition approaches~\cite{fortunato2010community,fortunato2016community}.
Despite some known limitations, such as an inherent resolution limit~\cite{fortunato2007resolution}, one of the most popular community detection algorithms is modularity maximisation, which seeks to optimize the modularity score proposed in a seminal paper by Newman and Girvan~\cite{newman2004finding}.

A current challenge in the analysis of real-world systems is that increased measurements and data availability have been creating a need for algorithms and analysis tools that scale to very large networks.
In this context graphons have emerged as one promising non-parametric generative model for the study of large networks (a more detailed introduction to graphons will be given in~\Cref{sec:preliminaries}; briefly, graphons are functions on the unit square originally proposed as continuous limiting objects for dense graph sequences~\cite{lovasz2006limits}).
Using the Aldous--Hoover representation theorem~\cite{jacobs2014unified,aldous1981representations,hoover1979relations}, it can be shown that graphons encapsulate a large number of popular existing generative graph models, such as the \emph{Erd\H os--R\'enyi graph}, the \emph{stochastic block model}~\cite{abbe2017community} and its variants, \emph{random dot product graphs}~\cite{athreya2017statistical}, and many further latent variable graph models~\cite{orbanz2014bayesian}.
Graphons are thus very flexible probabilistic models that can represent a wide range of network structures.

The price for this flexibility, however, is that a graphon-based network model may still be quite complex and not easy to interpret for a practitioner.
Hence, when inferring a graphon from empirical data, we might end up trading one large complex network for another complex object, which has arguably impeded the adoption of general graphon models by applied scientists (see~\cref{fig:privacyGraphon}a for an illustration).
Indeed, one reason for the interest in community detection is that communities enable a simplified description of a large network, by decomposing the network into modular ``building blocks''.
To address this issue in the context of graphons, in this work we develop community detection using a form of modularity optimization for graphons.
Our work thereby serves as a first step towards an effective, more interpretable summarization of an inferred graphon describing a complex network.

\subsection{Motivation: community detection for graphons}

\paragraph{Approximating and simplifying graphons via communities}
As outlined above, one main motivation for developing a form of community detection for graphons arises from the need to simplify an empirically obtained graphon further. As graphons encapsulate many popular generative graph models, obtaining a community structure for graphons allows to estimate community structure for those without the need to construct networks from them. While we here concentrate on providing an ``assortative block simplification'' in terms of community structure via modularity maximisation, other approaches are conceivable as well. For instance, we may use low-rank approximations or corresponding spectral embeddings to obtain a simplified description of a graphon.
In fact, we may be interested in other types of potential simplifications and analyses of graphons, including further structural decompositions such as core-periphery approximations~\cite{rombach2014core} or centrality measures~\cite{avella2018centrality}.
See also~\Cref{ssec:related_work} for a brief overview of the emerging area of graphon analysis for applications. 
 
\paragraph{Community detection via continuous optimization}
As graphons are objects defined in a continuous domain (functions on the unit square), some techniques from continuous mathematics become applicable for their analysis.
This may provide a potential reservoir of new algorithms and analysis tools for networks.
While not the main focus of this paper, we will see a simple example of this kind in \Cref{subsubsec:optimisation}, where an analytical solution to the modularity optimization problem on certain graphons is derived.
More generally, a fruitful endeavour could be to characterise how classical network algorithms might be seen as discrete approximations of certain problems defined on graphons, which could lead to a deeper understanding of these problems. 

\paragraph{Privacy-preserving computation}\label{sec:privacy}
The protection of sensitive data against unauthorised access is a crucial part of data warehousing and data analysis. 
Large-scale network data can be such sensitive information.
For instance, network data may comprise health records such as brain connectomics~\cite{hagmann2008mapping}, individuals' social contact data such as their Facebook network~\cite{maier2017cover}, or commercially confidential information. 
Facebook friendships, for example, may be used to expose an individual's sexual orientation~\cite{jernigan2009gaydar}.
Graphons are one way to represent such data as an approximation that partially preserves large- and meso-scale features of the data, while providing some anonymity for individuals. 
Importantly, it has been demonstrated that we can estimate graphons from network data while preserving the privacy of individual nodes~\cite{borgs2015private,borgs2018revealing}

By sharing such a privacy preserving graphon, a data collecting entity can thus enable further analysis of such network data.
This may lead to valuable insights into the system, while simultaneously preserving the privacy of the involved individuals.
\Cref{fig:privacyGraphon}b shows a schematic for such an approach for the case of community detection, in which an external entity is given access to a graphon created from a system graph $G$ in a privacy-preserving way. 
The external entity can now detect communities and analyse the graphon to gain insights about the system, without compromising the privacy of the individual nodes.
This procedure may also be of interest if the particular network analysis to be performed cannot be performed by the data collecting entity, e.g., because of a lack of sufficient computational resources.

  \begin{figure}[t]{
      \centering 
  \includegraphics[width=0.99\textwidth]{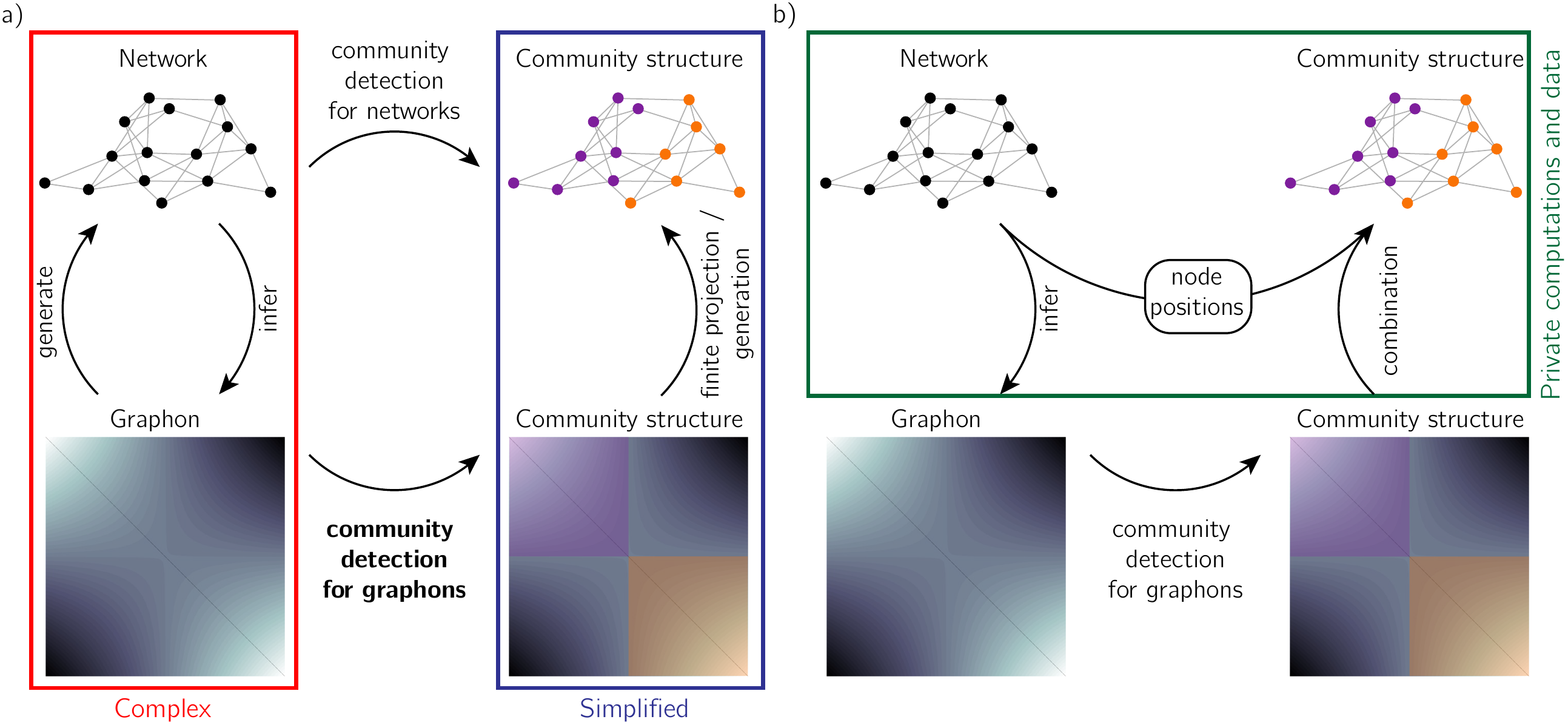}}
  \caption{\textbf{Schematic: graphon modularity enables the detection of community structure, which are---similar to networks---a simplified representation of the mesoscale connectivity structure. Furthermore, graphon-based computations enable privacy preserving network analysis.} (a) Classic community detection infers a community structure in a network. Networks can be generated from graphons. The inverse process is the inference of graphons from network data. We establish a community detection for graphons which allows the inference of communities in graphons.  (b) For the privacy-preserving network analysis, the data-hosting entity infers a graphon from a network and keeps the node positions $x_i$ private, while sharing the graphon with external entities (or publicly). The external entity can then compute the community structure of the graphon and returns it to the data-hosting entity, which combines it with the node positions $x_i$ to obtain the community membership of each node in the network. During the whole process, the privacy of individual nodes is preserved.
  }
  \label{fig:privacyGraphon}
\end{figure}

\subsection{Contributions}
In this paper, we introduce the problem of modularity optimization for graphons.
We define a modularity function for graphons and introduce corresponding algorithms to detect community structure in synthetic and empirically estimated graphons.
Specifically, we adapt the popular Louvain algorithm for modularity maximization in networks to the context of graphons.
Additionally, we discuss a simple continuous optimization variant for modularity we call \emph{sliced modularity optimization}.
For selected graphons, this enables us to derive analytical expressions of the optimal community structure. 
Our characterization of community structure in terms of the modularity of a graphon further provides a necessary and sufficient criterion for graphs to contain no modular structure.
Namely, for graphons that have a product form (the graphon operator has rank 1), any partition has modularity zero, i.e., there is no evidence for the presence of communities in terms of the modularity function.

As there already exist a number of algorithms for inferring a graphon from network data, we focus mainly on a scenario where the graphon of interest is already given when developing our modularity optimization approach. In~\Cref{subseq:empiricalExamples}, however, we also study graphons inferred from empirically observed and synthetic network data.
We report on how the graphon inference steps can impact the results of modularity optimization using numerical simulations.
Using stochastic blockmodels, we also illustrate that combining graphon estimation with graphon community detection may improve the performance of modularity-based community detection in comparison with standard, graph-based modularity maximisation.
Intuitively, by first fitting a graphon to the network we smooth out random fluctuations in the data, which can impede modularity maximization. At the same time, the full graphon description can be kept for a more refined interrogation of the network structure, e.g., in terms of centrality measures~\cite{avella2018centrality} or for sampling surrogate network data.

While in the setting of stochastic blockmodels, there is a planted ground truth partition to compare against, for real networks this is, of course, not the case.
We thus investigate numerically for several real networks the extent to which standard modularity optimization (i.e., directly applying modularity maximization on an observed graph) and graphon-based modularity maximization (i.e., first infer a graphon, then find the partition maximising graphon modularity) leads to the same results.
For empirical data, we find that the accord between the community structure detected in a network and its estimated graphon depends on the chosen estimation algorithm.  
Assuming that partitions that optimize modularity in the graph are those we want to detect in the graphon, this is an important aspect if we want to use the graphon based scheme for privacy preserving computation.

\newpage

\subsection{Related work}\label{ssec:related_work}
\paragraph{Graphon-based network analysis}
Our work relates to the development of graphon analysis tools and metrics, which is an active field of research.
Indeed, centrality measures~\cite{avella2018centrality}, controllability~\cite{gao2017control}, random-walk Laplacians~\cite{petit2019random}, epidemic spreading~\cite{vizuette2020sis}, and subgraph counts~\cite{coulson2016poisson} have all been recently derived for graphons.

Graphons have also been proposed as a general tool for encapsulating sensitive information in a node-private way~\cite{borgs2015private,borgs2018revealing}.
In light of our previous discussion, the graphon-modularity we present here, may thus also be interpreted as a privacy-preserving method for community detection. 

\paragraph{Graphon estimation}
The estimation of graphons from network data has been studied in a number of publications, e.g.,~\cite{wolfe2013nonparametric,olhede2014network,chatterjee2015matrix}. 
Approaches include stochastic blockmodel approximations~\cite{airoldi2013stochastic}, neighbourhood-smoothing algorithms~\cite{zhang2017estimating}, or sorting-and-smoothing algorithms~\cite{yang2014nonparametric}. 
Estimating a general graphon exactly is only possible under certain identifiability conditions~\cite{chan2014consistent}.
In practice, however, many approaches estimate a graphon as a \emph{mixture} of stochastic blockmodels, each consisting of many blocks of equal size (also called \emph{network histograms}~\cite{olhede2014network}).

\paragraph{Community detection in networks}
Many different approaches have been developed to detect communities in networks~\cite{fortunato2010community,fortunato2016community} and modularity-maximisation is one of the most widely-used paradigms. 
The {\sc Louvain}-algorithm is a fast heuristic to solve the modularity-maximisation problem~\cite{blondel2008fast}. 
An alternative to community detection is to identify an embedding of discrete nodes in a continuous latent space (e.g.,~\cite{grover2016node2vec}). 
These latent embeddings may be interpreted in terms of generalized community structure~\cite{newman2015generalized,hoff2002latent}, or indeed graphon estimation~\cite{orbanz2014bayesian}.

\subsection{Outline}\label{ssec:outline}
The rest of the paper is structured as follows. 
After briefly reviewing some preliminaries in \cref{sec:preliminaries}, we formally define graphon-modularity in \cref{sec:results}. 
We then present approaches to optimise the graphon-modularity in \cref{subsubsec:optimisation} and we discuss different synthetic graphons in \cref{subsec:syntheticGraphons}.
In~\cref{subseq:sampling} we consider the issue of modularity maximisation on graphon inferred from finite graphs and graphs sampled from latent graphons.
We finally perform numerical experiments on the performance of our methods for empirically networks in \cref{subseq:empiricalExamples} and conclude with a discussion in \Cref{sec:discussion}.

\section{Preliminaries}\label{sec:preliminaries}

\subsection{Graphs}
We represent networks mathematically as finite, unweighted, and undirected graphs.
A graph is an ordered pair $G=(V,E)$ composed of a set $V$ of nodes (vertices), and a set of edges (links) $E$, where each edge $e\in E$ corresponds to an unordered tuple of two nodes~\cite{newman2010networks,agnarsson2006graph}. 
We denote the number of nodes in a graph as  $N=|V|$ and the number of edges as $M=|E|$. 
Without loss of generality we will assume that the nodes of the graphs have been labeled by the nonzero integers such that $V = \{1,\ldots, N\}$.
Using this labeling we can define the \emph{adjacency matrix} $\mathbf{A}$ of the (labeled) graph as the $N \times N$ matrix with elements $A_{ij}=1$, if $(i,j) \in E$ and $A_{ij}=0$, otherwise. 
Based on this algebraic representation we can compute the degree $k_i$ of each node $i$ as $k_i=\sum_{j=1}^N A_{ij}$, i.e., $k_i$ is the number of edges incident to node $i$.

\subsection{Graphons}
Graphons originally emerged in the study of limits of large scale networks~\cite{Lovasz2006,borgs2008convergent,borgs2017graphons,lovasz2012large,glasscock2015graphon}, and have been defined as limits of (dense) graphs for which the number of nodes $N\rightarrow \infty$. 
A \emph{graphon} is a measurable function $W :[0,1]^{2} \rightarrow[0,1]$, that is symmetric with respect to its arguments such that $W(x, y)=W(y, x)$. 
An intuitive way to think of a graphon is in terms of the limiting object of a heatmap image (``pixel picture'')~\cite{glasscock2015graphon} of a graph's adjacency matrix as follows.
We assume that as the graph size $N\rightarrow \infty$, the heatmap image (the pixel picture) of the adjacency matrix is always spatially scaled to maintain the dimension of the unit square.
In the limit there are thus $N\rightarrow \infty$ nodes associated with the unit interval $[0,1]$ and the values $x$ and $y$ in a graphon may thus be interpreted as the indices of the vertices in an infinite graph. 

While this heuristic explanation provides some intuition, it needs refinement. 
Observe that for any graph we can permute the node labels and thereby change its representation in terms of the adjacency matrix, while leaving the graph structure unchanged, 
For a graphon to be a valid limiting objects of graphs rather than of adjacency matrices (i.e., labeled graphs), a graphon can only be defined as a limiting object up to measure preserving bijections of its arguments, i.e., $W(x,y) = W(\pi(x),\pi(y))$, where $\pi: [0,1] \rightarrow [0,1]$ is a measure preserving map.
A more precise characterization of the equivalence classes of graphons is provided in~\cite{lovasz2012large,glasscock2015graphon}.

Graphons may also be interpreted as nonparametric random graph models, as introduced in \cite{Lovasz2006} under the name $W$-random graphs. 
We can sample a random graph of size $N$ within this model as follows.
First, each node $i\in \{1,\ldots,N\}$ is assigned a latent position $u_i \in [0,1]$ (typically drawn uniformly at random).
Second, any pair of nodes $i,j$ is then connected with an edge with probability $\mathbb{P}(A_{ij} = 1) = W(u_i,u_j)$.

Similar to graphs we may associate a degree function $k(x)$ to every node $x$ in a graphon via the following Lebesgue integral:
\begin{align}
	k(x) = \int_0^1 W(x,y)\ dy\,.
\end{align}
Likewise we define the edge density $\mu$ of a graphon as:
\begin{align}
    \mu = \int_{{[0,1]}^2} W(x,y)\ dxdy =  \int_0^1 k(x)\ dx\,.	
\end{align}

\subsection{Community-detection in networks}
A community is a set of nodes, such that nodes within the same community are more densely connected to each than to nodes in other communities~\cite{fortunato2010community,fortunato2016community,porter2009communities}.
We restrict our discussion to \emph{non-overlapping} communities, such that each node belongs to exactly one community. 
For convenience, we describe the vertex to group assignment by the function $g_V: V \to \{1,2,\dots,c\}$, which maps each node to one of the $c$ communities. 

Many different heuristic algorithms to find such a function have been developed (see~\cite{schaub2017many,yang2016comparative,rosvall2019different} for reviews). 
Among the most widely-used heuristics is the so-called \emph{modularity maximisation}. 
For this, one defines a modularity function
\begin{align}\label{eq:modularity}
    Q(g_V) = \frac{1}{2M} \sum _{i,j=1}^N \left( A_{ij} - P_{ij}\right)\delta[ g_V(i),g_V(j) ] = \frac{1}{2M} \sum _{i,j=1}^N B_{ij} \, \delta[ g_V(i),g_V(j) ]\,,
\end{align}
which is a quality index for a group assignment function $g_V$ of a network with adjacency matrix $\mathbf{A}$~\cite{newman2006modularity}.
Here, the matrix $\mathbf{B}=[B_{ij}]$ is the \emph{modularity matrix} and its entries $B_{ij} = A_{ij}-P_{ij}$ are equal to the entries of the adjacency matrix shifted by a chosen \emph{null model} term $P_{ij}$.
This null model term is typically chosen to be the expected connection strength between nodes $i$ and $j$ under a chosen random graph model.
The null model thus serves as a baseline to which the actual connections of the adjacency matrix are compared.
Since the Kronecker-delta $\delta[ \cdot, \cdot ]$ in \Cref{eq:modularity} is $1$ if its arguments are equal and $0$ otherwise, the above sum only takes into accounts elements of the modularity matrix for which nodes $i$ and $j$ belong to the same community $g_V(i)=g_V(j)$. 
Accordingly, the modularity of a particular group assignment $g_V$ is equal to the sum (rescaled by $2M$) of the intra-community edges minus the expected weight of intra-community edges.

There are numerous choices for the null model term $P_{ij}$.
There exist null models for spatially embedded networks~\cite{expert2011uncovering,sarzynska2015null}, null models for networks constructed from correlation models~\cite{macmahon2013community,bazzi2016community}, and for many other situations~\cite{fortunato2016community}.
Here we focus on the typically considered Newman--Girvan null model $P_{ij}=k_i k_j /(2M)$~\cite{newman2004finding}, also known as the Chung--Lu model~\cite{Chung2002}, which preserves the expected degree distribution of the graph.

With this choice for $P_{ij}$ the modularity $Q$ can be written as:
\begin{align}\label{eqn:modularityGen}
    Q(g_V) = \frac{1}{2M} \sum _{i,j=1}^N \left( A_{ij} - \frac{k_i k_j}{2M} \right)\delta[g_V(i),g_V(j)] \,.
\end{align}
Clearly there are a number of equivalent group assignments, as the group labels can be permuted without changing the induced partition of the nodes.
It is thus the partition of the nodes induced by the group labels that is important for the modularity score, rather than the labels per se.

The task of community detection can now be formalized as the following~\emph{modularity-maximisation problem}~\cite{bazzi2016community}:
Find a partition of the nodes (respectively a group assignment) that maximises the modularity function $Q$
\begin{align}\label{eq:mod_optimization}
    \max_{g_V} Q(g_V) = \max_{g_V} \left\{\sum _{i,j=1}^N B_{ij} \, \delta[ g_V(i),g_V(j) ] \right\}\,, 
\end{align}
where $g_V$ is a group assignment function, mapping each node to one community.

Note that in the optimization problem \Cref{eq:mod_optimization} the number $c$ of communities is not fixed. 
The modularity maximization problem may thus be viewed as searching over the set of all possible partitions of the nodes.
Since this set becomes extremely large even for moderately sized graphs, \Cref{eq:mod_optimization} is computationally difficult to optimize.
In fact, it has been shown that modularity maximisation is an NP-hard problem~\cite{brandes2007modularity}.
In practice, the modularity optimization problem is thus solved approximately using (greedy) heuristic procedures, such as the \emph{Louvain algorithm} or the \emph{Leiden algorithm}~\cite{blondel2008fast,genlouvain,traag2019louvain}, which have been shown to yield good empirical performance.

\section{A modularity function for graphons}
\label{sec:results}

In this section, we define the modularity function for graphons, which we will later employ for community detection in graphons.
Analogously to the community-detection problem for graphs, the community-detection problem for graphons can be expressed as the identification of a group assignment function
\begin{align}
	g: [0,1] \to \{1,2,\dots,c\}\,,
	\label{eqn:communityFunction}
\end{align}
which assigns each node position $x \in [0,1]$ to one of $c$ communities.
To find such a community-assignment function, we define a modularity function for graphons.

\begin{definition}[Graphon modularity]\label{def:graphonModularity}
For a graphon $W(x,y)$, a graphon null model $P(x,y)$ and a group assignment function $g$, we define the \emph{graphon-modularity} as
\begin{align}
    Q(g) = \frac{1}{\mu} \int_0^1 \int_0^1  \underbrace{\big[ W(x,y) - P(x,y)\big] }_{\mathrm{modularity\ surface}\ B(x,y)}\delta(g(x)-g(y))\ dx dy\,,
\label{eqn:modularityGraphon}
\end{align}
where $\delta(\cdot)$ denotes Dirac's delta function, and we have defined the \emph{modularity surface} $B(x,y) = W(x,y) - P(x,y)$ as analog of the modularity matrix for graphons.
\end{definition}

Analogous to graph case, the graphon modularity function may be interpreted as a measure of the quality of a node partition induced by the group assignment function $g$.
Likewise, the modularity surface $B(x,y)$ indicates how well a graphon $W(x,y)$ at location $(x,y)$ is connected compared to the null model term $P(x,y)$.

As for graphs, different choices for the \emph{null model} $P(x,y)$ may be sensible. 
For simplicity we here restrict the discussion again to a Newman--Girvan-type null model:
\begin{align}
	P(x,y) = \frac{1}{\mu} k(x)k(y)\,.
\end{align}

Note that since graphs generated for sufficiently smooth graphons~\cite{lovasz2006limits,avella2018centrality} will converge to graphons in the limit, the definition of the modularity surface~\cref{def:graphonModularity}, precisely corresponds to the (scaled) limiting object of the modularity matrix.
\begin{proof}(Sketch)
    It can be shown that the normalized degree of node $i$ converges as $N\rightarrow \infty$~\cite{avella2018centrality} and accordingly the degree density of the graph will converge to the graphon density $\mu$.
    This implies that the null model term will be well defined in the limit.
    Since furthermore the (scaled) adjacency matrix $\mathbf{A}/N$ will converge to the graphon $W$, both terms that define the modularity surface converge and are well defined.
\end{proof}

Given a group assignment $g(x)$, we further define the (relative) size $S_{i} \in (0,1]$ of a community $c_i$ as
\begin{align}
    S_{i} = \int_0^1 \delta (g(x),c) \ dx\,.
\end{align}
The maximal size $S_{i}$ of any community equals one, which indicates that there exists a single group consisting of all nodes. 

Analogously to graphs, we can now detect communities in graphons by finding a function $g$ that maximises the graphon modularity \cref{eqn:modularityGraphon}.
\begin{definition}[Modularity-maximisation problem for graphons]\label{def:graphonModularityMaximisation}
    Given a non-empty graphon $W(x,y)$ and a graphon null model $P(x,y)$ that define the modularity surface ${B(x,y)=W(x,y)-P(x,y)}$, the \emph{modularity-maximisation problem} is:
\begin{align}
    \max_{g(x)} & \int_0^1 \int_0^1 B(x,y)\delta(g(x),g(y))\ dxdy\,,
\end{align}
where $g$ is a group assignment function such that $S_{i} > 0$ for all communities $c_i$, i.e., each community has a nonzero measure.
\end{definition}

\begin{remark}[Measure zero sets and $L_2$ equivalence]
    Recall how a graphon can only be defined meaningfully up to measure preserving transformations.
    The same is true, \emph{mutatis mutandis}, for the group assignment function $g$.
    Indeed, from the above \cref{def:graphonModularity} of graphon modularity, it should be clear that a group assignment function can only be meaningfully defined up to $L_2$ equivalence.
    Specifically, let $L_2([0,1])$ denote the Hilbert space of functions $f:[0,1]\rightarrow \mathbb{R}$ with inner product $\langle f_1, f_2\rangle = \int_0^1f_1(x)f_2(x)dx$ and norm $\|f_1\|=\sqrt{\langle f_1,f_1\rangle}$.
    The elements of $L_2([0,1])$ are the equivalence classes of integrable functions that differ only on measure zero sets, i.e., we identify two function $f_1\equiv f_2$ with each other if $\|f_1-f_2\|=0$.

    For instance, changing the group assignment value $g(x)$ for any single $x\in [0,1]$ will not alter the modularity $Q(g)$ or the size of the communities $S_i$.
    This measure preserving change of $g$ leads to a non-identfiability of the precise function $g$ in the optimization of graphon modularity.
    However, analogous to the possibility of permuting the group labels, this non-identifiability does not lead to practical problems, as in practice we are only concerned with group assignment functions up to $L_2$ equivalence.
    Similarly, we only consider communities with nonzero measure, i.e, we will require that $S_{i}>0$ for all $c_i$, as specified in \cref{def:graphonModularityMaximisation}.
\end{remark}

As for modularity optimization for graphs, finding an optimal solution for the graphon modularity optimization problem is only possible for special cases. 
In \Cref{subsubsec:optimisation}, we thus explore two procedures to find either analytical expressions of the optimal community structure, or approximate solutions via numerical algorithms.

\section{Optimising graphon-modularity}\label{subsubsec:optimisation}
Here we explore two approaches to identify a maximum-modularity partition for graphons (see \Cref{fig:graphonOptimisation} for schematic representations). 
The first approach is to discretise the modularity surface $B(x,y)$ (see \Cref{def:graphonModularity}) and use a generalised Louvain algorithm ({\sc GenLouvain}) to find a group assignment function $g$.
The second approach is (semi-)analytical and works if we can constrain the set of group assignment function $g$ to be monotonically increasing, which we can do for certain synthetic graphons.
In this case we can find the exact maxima of modularity.
We will show in \Cref{subsec:syntheticGraphons} that for selected synthetic graphons both methods return essentially identical results, providing some further validation for the Louvain heuristic.

\subsection{Maximisation of graphon modularity via discretization}
We use the {\sc GenLouvain} algorithm~\cite{genlouvain} on a piecewise constant approximation of $W(x,y)$, as illustrated in \Cref{fig:graphonOptimisation}, to heuristically optimize graphon modularity.
{\sc GenLouvain} is a variant of the fast \emph{Louvain}-algorithm~\cite{blondel2008fast}, which was originally designed for standard modularity optimization on simple graphs.
To apply the \textsc{GenLouvain} algorithm, we first need to discretize the graphon $W(x,y)$ appropriately by trading off two aspects.
First, we need to chose a fine enough grid to capture the variation of the graphon $W(x,y)$ in $x$ and $y$ direction.
Second, we would like to choose an as coarse as possible grid, to limit the computational costs of the optimization performed via {\sc  GenLouvain}.

In the following we approximate $W(x,y)$ using a uniformly spaced grid of size $2000\times 2000$ unless otherwise stated. 
For this grid size, detecting community structure in discretised graphons is possible with commodity hardware in a few minutes.
We observe that this choice of the discretization for $W(x,y)$ is fine enough for all the problems considered in this paper.
Increasing the resolution of the grid further has essentially no effect for the results, in practice.

We remark, that more elaborate discretization schemes are conceivable and might result in computational gains.
For instance, one could employ multigrid discretization schemes to obtain a better approximation of local features of the modularity surface, without incurring a large extra computational cost~(for a review of multigrid schemes see~\cite{stuben2001review}).

\begin{figure}[t]
\begin{center}
\includegraphics[width=0.49\textwidth]{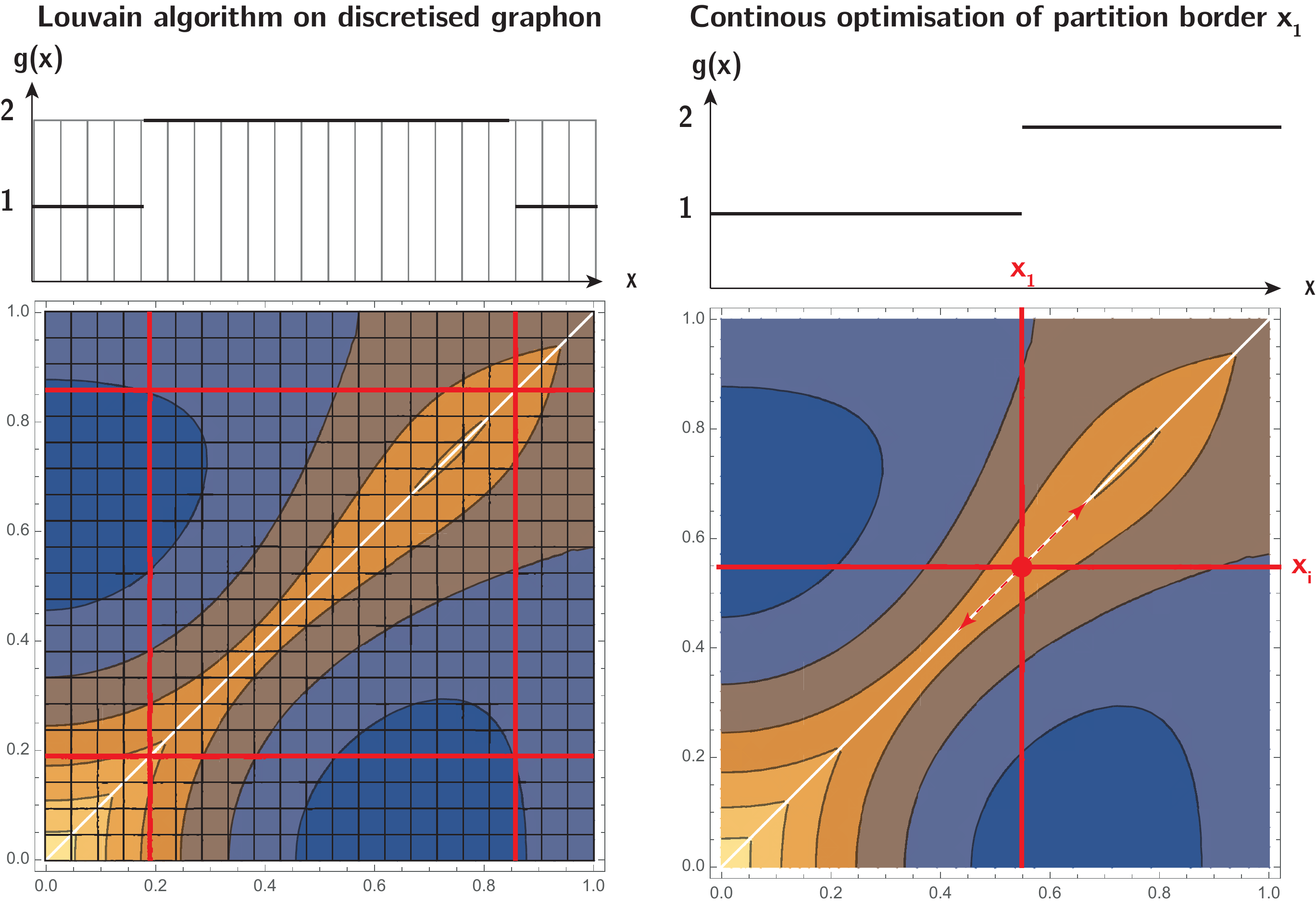}
\end{center}
\caption{\textbf{Optimizing graphon modularity.} We discuss two procedures to maximise graphon-modularity. (Left panel) We discretise the graphon into intervals $[x,x +\Delta x]$ and apply a Louvain algorithm to optimise the modularity function for the resulting matrix. (Right panel) We restrict the modules to be continuous intervals (in this example $[0,x_1]$ and $(x_1,1]$). We optimise the modularity by continuously varying the border $x_1$ between the communities.
  In this schematic, there exist only two communities, although in general there can be any number of communities.}
  \label{fig:graphonOptimisation}
\end{figure}

\subsection{Maximisation of graphon modularity via continuous optimisation}
We now consider a setup in which we can analytically establish the optimal partitions of a graphon in terms of the modularity function.
This provides us with a way to validate the results we obtain from the discretization based graphon-modularity optimization outlined in the previous section.
Moreover, it highlights that (within certain situations) tools from continuous optimization can be employed to analyse modularity maximizing partitions of graphons.

To this end, we restrict the possible group assignment functions $g(x)$ to be piecewise constant on $c$ intervals:
\begin{equation}
    g(x) = 
    \begin{cases}
        1 \quad \text{if} \quad x\in[0,x_1)\\
        2 \quad \text{if} \quad x\in[0,x_1)\\
        \vdots \\
        c \quad \text{if} \quad x\in[x_{c-1},1]\\
    \end{cases}
\end{equation}
For certain synthetic graphons the partitions can be assumed to be of such a form.

If the group assignment function $g(x)$ is of such a form, we can rewrite the graphon modularity function as:
\begin{align}
	Q(g) = \frac{1}{\mu} &\sum_{i=1}^c L(x_{i-1},x_i) \,,\ \text{with} \\
	L(a,b) =& \int_{a}^{b} \int_{a}^{b} B(x,y) dx dy = 2\int_{a}^{b} \int_{a}^{x}  B(x,y)dxdy\,,
\end{align}
where the last equality follows from the symmetry ($B(x,y)=B(y,x)$) of the modularity surface.
We call the function $L(a,b)$ a \emph{modularity slice} as the modularity function may be seen as a simple linear sum of these slices.
Note that within the above formulation, we have thus restricted the (combinatorial) optimization problem to a much simpler optimization with $c-1$ degrees of freedom, namely finding the end-points of the first $c-1$ intervals over which $g(x)$ is constant.
For instance, if we know that there are $c=3$ communities we obtain:
\begin{align}
	Q(g) \propto L(0,x_1) + L(x_1,x_2) + L(x_2,1)\,,
\end{align}
which has only $x_1$ and $x_2$ as free variables.
This continuous optimisation problem can in many cases be solved analytically. 
When this is not directly possible, however, we can optimise it using standard optimisation procedures, such as the Nelder--Mead method~\cite{nelder1965simplex}.
In the following, we will use the sliced-modularity approach to prove the maximum-modularity partition for a synthetic graphon. In the appendix, we discuss a generalised sliced-modularity approach, which allows the detection of communities in settings when less is known about the location of communities on the line.

\section{Modularity optimization for synthetic graphons}\label{subsec:syntheticGraphons}

In this section we explore the detection of community structure for different given synthetic graphons, and show how the obtained group assignments provide some simplified description block description of the graphons.

\subsection{Graphons with zero modularity surface}\label{subsec:withoutCommunity}
Maximising the graphon-modularity returns a group assignment function $g(x)$ with the highest graphon-modularity $Q$. 
If the modularity surface  of the graphon is $B(x,y)=0$, however, all functions $g(x)$ have the same modularity $Q=0$.
In this case, it is not possible to find a partition that has a higher modularity than any other.
Accordingly, we will say that such a graphon does not exhibit a community structure. 
Note that in contrast to the graph case, such a degenerate situation is possible even if the graphon itself is non-zero and well defined.

\begin{proposition}
Graphons of the form $W(x,y)=f(x)f(y)$ do not have community structure.
\label{theorem:multiplicative}
\end{proposition}

\begin{proof}
Let $W(x)$ be a graphon of the form $W(x,y)=f(x)f(y)$. 
The degree function is thus $k(x) = f(x) \int_0^1 f(y) dy$, the edge density is $\mu = \int_0^1 f(y) dy\int_0^1 f(x) dx$ and the null model term can be computed as $P(x,y) = f(x)f(y)$. 
Therefore the modularity surface is $B(x,y)=0$  and it follows that the modularity function is zero for all partitions.
\end{proof}

An example of a graphon without community structure is the graphon associated to the Erd\H os--R\'enyi (ER) graph model $G(n,p)$, which can be represented as a graphon $W(x,y)=p$ with a constant connection probability $p \in [0,1]$. 
Indeed the ER graphon belongs to the class of multiplicative graphons of~\Cref{theorem:multiplicative} as $W(x,y)=\sqrt{p}\sqrt{p} = f(x)f(y)$ and thus has no community structure.
In fact, the reverse is also true and all non-empty graphons without community structure are multiplicative (up to $L_2$ equivalence).

\begin{proposition}
Let $W(x,y)$ be a non-empty graphon with vanishing modularity surface $B(x,y)=0$. Then, the graphon is of the form $W(x,y)=f(x)f(y)$.
\label{theorem:zero}
\end{proposition}

\begin{proof}
    As $W(x,y)$ and $P(x,y)$ are positive, if follows from $B(x,y)=0$ that $W(x,y)\equiv P(x,y)$ (in the sense of $L_2$ equivalence), where the null model term is given by $P(x,y)=k(x)k(y)/\mu$. 
    We can therefore write $W(x,y)=f(x)f(y)$ with $f(x)=k(x)/\sqrt{\mu}$. 
\end{proof}
Note that this implies that the linear (graphon) integral operator, i.e., the integral operator whose kernel is given by the graphon has rank 1

\subsection{A core-periphery model: the $\lambda$-graphon}\label{subseq:lambdaGraphon}

\begin{figure}[b!]
{\centering 
\includegraphics[width=0.99\textwidth]{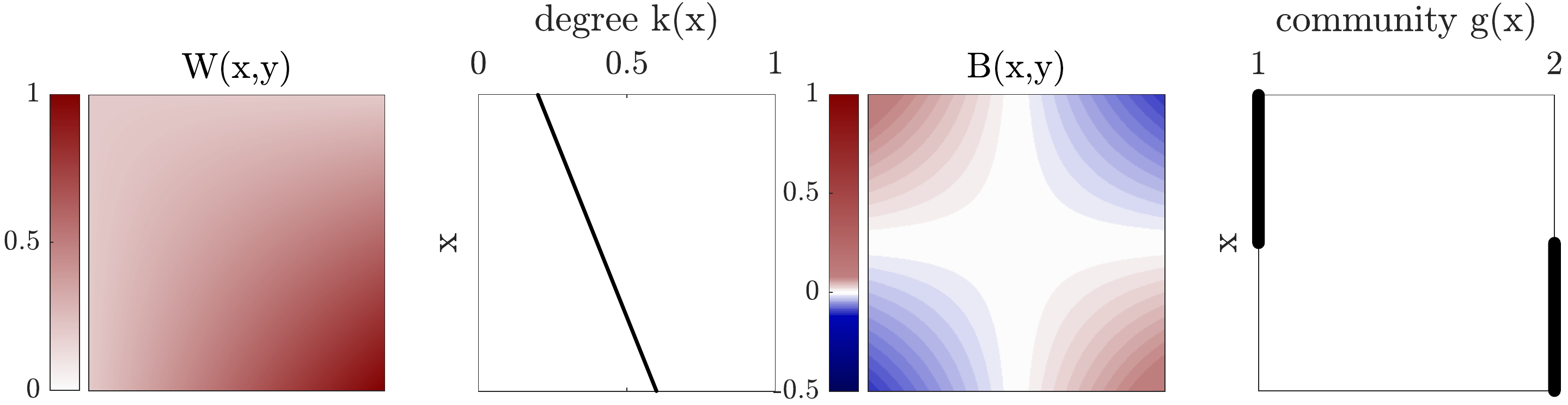}
\caption{\textbf{Graphon modulurity maximization for a graphon with core-periphery structure.} In a core-periphery $\lambda$-graphon (see text) with $\lambda = 0.2$ we detect two communities. The panels show the graphon $W(x,y)$, the degree $k(x)$, the modularity surface $B(x,y)$, and the community structure $g(x)$ returned by the GenLouvain algorithm, respectively. The sliced-modularity approach returns similar results for the community structure.}
\label{fig:lambdaGraphon}}
\end{figure}

We consider a graphon with core-periphery structure $W(x,y) = (1-\lambda)xy + \lambda$, modulated by a parameter $\lambda \in [0,1]$. 
We will call this particular graphon $W(x,y)$ the $\lambda$-graphon in the following.
Note that the $\lambda$-graphon may be seen as a convex mixture of (a) a flat connectivity profile with probability $\lambda$, corresponding to random ER-like connectivity, and (b) a coordinate dependent connectivity profile $xy$, which may be interpreted in terms of a continuous core-periphery structure.
Indeed, the larger the coordinate of the node, the higher its connectivity, as computing the degree function of the $\lambda$-graphon confirms:
\begin{equation*}
    k(x) = \lambda +  (1-\lambda)x/2.
\end{equation*}
The edge density of the $\lambda$-graphon is $\mu = (1-\lambda)/4 + \lambda$ and, accordingly, the modularity surface is
\begin{equation*}
    B(x,y)=(1-\lambda)\lambda \frac{(2x-1)(2y -1)}{(1+3\lambda)}
\end{equation*}

Note that when $\lambda=0$ and $\lambda=1$ we have $B(x,y) =0$ and we cannot detect communities for these cases as discussed above.
For $\lambda=1$, the $\lambda$-graphon becomes the constant graphon $W(x,y)=1$, which clearly does not have a community structure as all nodes are equivalent.
For $\lambda =0$ we obtain $W(x,y)=xy$, which does not have a community structure because of its multiplicative structure.
Numerically optimising the graphon modularity with the {\sc GenLouvain}-approach for $\lambda=0$ and $\lambda=1$ yields indeed a single community with $Q=0$.

Let us now  focus on scenarios for which $\lambda \in (0,1)$. 
In Fig.~\ref{fig:lambdaGraphon}, we show the degree $k(x)$, the modularity surface $B(x,y)$, and the detected community structure $g(x)$ for $\lambda=0.2$ as an example.
Optimising modularity with the {\sc GenLouvain}-approach for general $\lambda \in (0,1)$, we always obtain a partition into two communities that split the unit interval in half at $x=1/2$. 
To confirm this numerical results let us use the sliced-modularity approach to analytically compute the optimal community structure for two continuous groups.
First we compute the modularity slice:
\begin{align}
	L(a,b) = \frac{(a-b)^2 (-1+a+b)^2 (1- \lambda)\lambda }{3\lambda + 1}\,.
\end{align}
To obtain the optimal border $x_1$ between the two communities we maximise 
\begin{align}
	L(0,x_1;\lambda)+L(x_1,1;\lambda) = \frac{2(1-\lambda)\lambda}{1+3\lambda} (x_1-1)^2 x_1^2 = \kappa(x_1-1)^2 x_1^2 \,.
	\label{eqn:lambdaL}
\end{align}
As $\lambda \in (0,1)$, the factor $\kappa =2(1-\lambda)\lambda/(1+3\lambda)$ is greater than zero, and maximising the modularity is therefore equivalent to maximising $h(x) = (x_1-1)^2 x_1^2$. 
The local extrema of $h(x_1)$ with $dh/dx_1=0$ are $0$, $1/2$, and $1$. 
Evaluating the second derivative reveals that $0$ and $1$ are local minima and that $x=1/2$ is a maximum. 
For $\lambda\in (0,1)$ the optimal modular structure thus indeed consists of two equal-sized communities with the border at $x_1=1/2$, confirming the numerical results from the \textsc{GenLouvain} approach.

Using the above information about the optimal split for $\lambda \in (0,1)$, we can also calculate the maximum modularity as a function of the $\lambda$ parameter:
\begin{align}
Q_{\text{max}}(\lambda) =  \frac{L(0,1/2;\lambda)+L(1/2,1;\lambda)}{\mu(\lambda)} = \frac{\lambda (1-\lambda)}{2(3\lambda +1)^2}\,,
\end{align}
which has its maximum $Q_{\text{max}}=1/32=0.03125$ at $\lambda =1/5$. 
The example shown in Fig.~\ref{fig:lambdaGraphon} is therefore the $\lambda$-graphon with the largest possible modularity score.
Note that as $\lim _{\lambda \to 0} Q_{\text{max}}(\lambda) = 0$ and  $\lim _{\lambda \to 1} Q_{\text{max}}(\lambda) = 0$, the maximum modularity $Q_{\text{max}}$ vanishes when the $\lambda$-graphons  approach the boundary cases without community structure.
In this case the optimal value of modularity is thus continuous in $\lambda$.

\subsection{A uniform attachment model: the $\max$ graphon}\label{subseq:maxGraphon}

\begin{figure}[b!]
{\centering \includegraphics[width=0.99\textwidth]{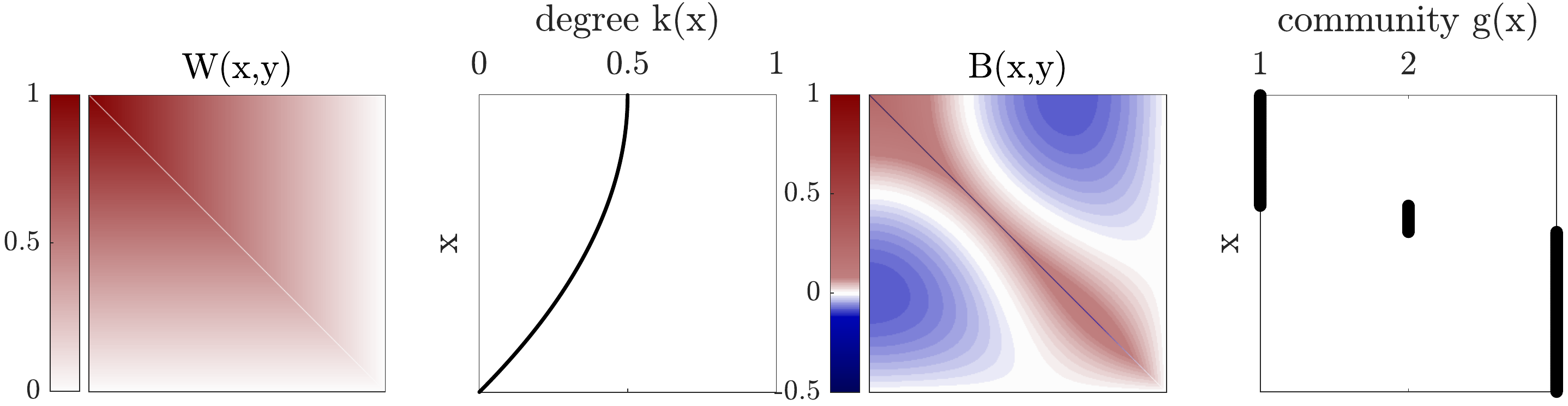}}

\caption{\textbf{Graphon modularity maximisation for a graphon with uniform attachment.} For the max graphon discussed in the text, we obtain a partition consisting of three communities. The panels show the graphon $W(x,y)$, the degree $k(x)$, the modularity surface $B(x,y)$, and the community structure $g(x)$ returned by the GenLouvain algorithm, respectively. The sliced-modularity apporach returns similar results for the community structure.}
  \label{fig:Maxgraphon}
\end{figure}

Synthetic models for growing networks have been widely used to model phenomena observed in real-world networks, such as long-tailed degree distributions~\cite{albert2002statistical}. 
These models typically consist of an iterative procedure that adds nodes and edges sequentially until a certain number of nodes is reached.

Some graphons can be seen as a limit of such a network growing process~\cite{borgs2011limits}. 
The sequence of \emph{growing uniform attachment} graphs we discuss next is a particular example of such a growth process that exhibits such a convergence.
We start with a graph $G_1$ consisting of a single node and no edges.
For $n\geq2$, we then construct $G_n$ from $G_{n-1}$ by adding a new vertex and adding every possible not already present edge in the network with probability $1/n$.
It can be shown~\cite{borgs2011limits} that this graph sequence almost surely converges to the \emph{max graphon} displayed in Fig.~\ref{fig:Maxgraphon}, which is defined as
\begin{align}
	W(x,y) &=  1- \max(x,y)
\end{align}

Computing the degree function of the $\max$ graphon as $k(x) = (1-x^2)/2$, we obtain a modularity surface of
\begin{align}
	B(x,y) = 1- \max(x,y) - \frac{3}{4}(x^2-1)(y^2-1)\,.
\end{align}
Optimising this function with the {\sc GenLouvain} approach yields a partition of the unit interval into the sets $[0,0.3730)$, $[0.3730,0.4605)$, and $[0.4605,1]$, corresponding 
to a medium sized community that consists of the highest-degree nodes, a small community consisting of medium degree nodes, and a large community consisting of all small degree nodes.
As we can see form the modularity surface in \Cref{fig:Maxgraphon}, these three sets provide indeed an approximate block-based description of the nodes within the graphon and may facilitate an easier interpretation of this model.

To confirm these results we again resort to our sliced-modularity approach using an Ansatz with three communities.
We obtain a modularity slice $L(a,b)$ that is a polynomial of order six in $a$ and $b$.
To estimate the maxima of the sixth order polynomial in terms of the partition endpoints $x_1, x_2$ we use Mathematica's Brent--Dekker method~\cite{brent1971algorithm} and obtain $x_1 \approx 0.369$ and $x_2 \approx 0.463$ (see online material). With higher computational effort, it is also possible to obtain the optimal community borders with the sliced-modularity approach for a larger number $c$ of communities. For $c=5$, for example, we obtain $x_1 = 0$, $x_2 \approx 0.369$, $x_3 \approx 0.463$, and $x_4 = 1$, which represents a community structure with two communities of vanishing size, yielding \emph{de facto}, the same community structure with $c=3$.
These results are well in line with the values obtained with {\sc GenLouvain}, given that the discretisation of the graphon for the {\sc GenLouvain} approach and possible numerical inaccuracies encountered when finding the maxima using the sliced-modularity method.

\section{Modularity of graphs sampled from unstructured and structured graphons}\label{subseq:sampling}

\subsection{Modularity of graphs sampled from graphons without structure}\label{subsec:samplingER}

In a number of benchmark tests for community detection on graphs the characterisation of a graphon with no community structure in~\Cref{subsec:withoutCommunity} is implicitly assumed to hold also for finite graphs sampled from a corresponding graphon model.
For instance, in \cite{fortunato2016community,Lancichinetti2009}, it is advocated that for graphs sampled from an ER model, a community detection algorithm should return a null result that indicates that there are no communities present.
While the idea is intuitively appealing, there are some issues when adopting this viewpoint in the finite regime.

In particular, a finite network sampled from a sparse ER model may not necessarily be representative of the underlying ER model, i.e., the random samples may not be concentrated around the expected featureless ER graphon~\cite{le2017concentration,joseph2016impact}.
This is in accordance with earlier results that state that samples drawn from an ER model can have a weak community structure arising from statistical fluctuations~\cite{guimera2004modularity}. 
In particular in~\cite{guimera2004modularity}, the authors showed that for large networks the maximum modularity approaches $Q \sim (pN)^{-2/3}$ (for a constant $p$, independent on $N$).
Therefore the modularity vanishes for $N\to \infty$, which matches the infinite size limit that graphons represent. 


\subsection{Modularity of graphs sampled from structured graphons}\label{subsec:pp}

\subsubsection{The planted partition model}
The planted partition (PP) model is a prominent random graph model with group structure, which can be described as a graphon.
For a $c$-block PP model each node $x$ belongs to exactly one group $g^{\star}(x)$, where $g^{\star}: [0,1] \rightarrow \{1,\ldots,c\}$ is the group assignment function.
 The $c$-block PPM graphon can then be written as:
\begin{align}
W(x,y) &= 
	\begin{cases}
        p_{\mathrm{in}} & \mathrm{if}\ g^{\star}(x) = g^{\star}(y)  \,,\\
        p_{\mathrm{ex}} & \mathrm{otherwise}, 
	\end{cases}
\end{align}
which means that nodes have an internal connection probability $p_{\mathrm{in}} \in [0,1]$ if they are in the same group, and an external connection probability $p_{\mathrm{ex}} \in [0,1]$ otherwise. 
To obtain an assortative community structure we assume $p_{\mathrm{ex}}<p_{\mathrm{in}}$. 

\begin{figure}[t]
{\centering \includegraphics[width=0.99\textwidth]{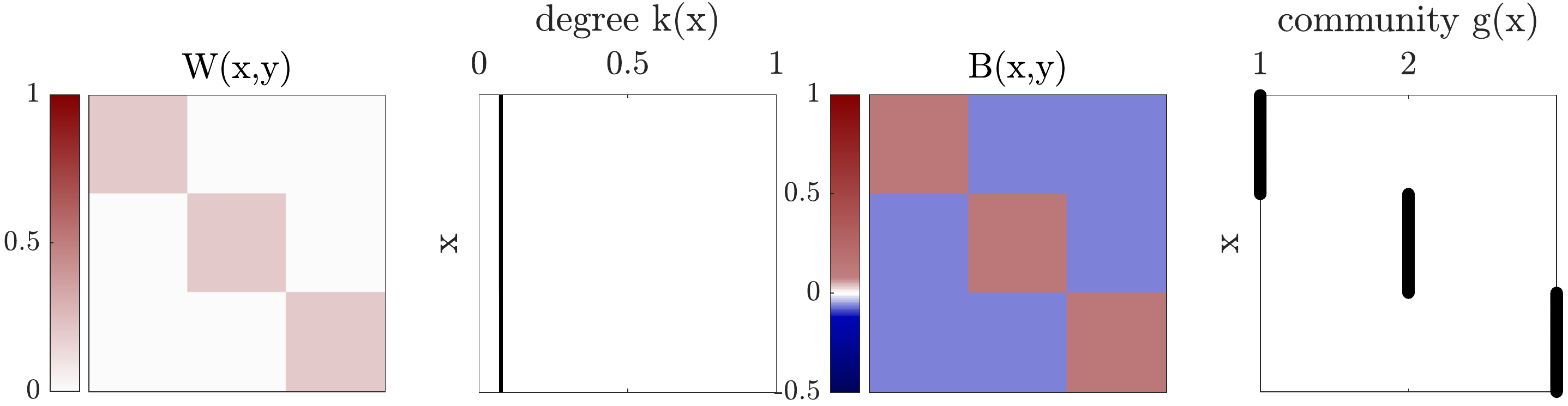}

    \caption{\textbf{Graphon modularity maximization for planted partition graphons.} In a planted partition graphon with $c=3$ blocks we recover the planted community structure. Each community has an internal connection probability of $p_{\mathrm{in}}=0.2$ and the connection probability between communities is $p_{\mathrm{ex}}=0.01$. The degree function is constant $k(x) = 1/3\times 0.2 + 2/3\times 0.01 \approx 0.0734$. The modularity surface $B(x,y)$ is positive in the communities and negative between communities. The community function $g(x)$ indicates that we correctly identify the three planted communities.}
  }
  \label{fig:SBMgraphonK3}
\end{figure}

For the PP graphon, we can compute the degree as: 
\begin{align}
	k_{\mathrm{PP}}(x) &= \frac{(p_{\mathrm{in}} + (c-1)p_{\mathrm{ex}}) }{c} \,,
\end{align}
and the total connectivity:
\begin{align}
	\mu_{\mathrm{PP}} &=  \frac{(p_{\mathrm{in}} + (c-1)p_{\mathrm{ex}}) }{c}\,,
\end{align}
Accordingly, the modularity surface is
\begin{align}
	B_{\mathrm{PP}}(x,y) &= 
	\begin{cases}
        \frac{c-1}{c} (p_{\mathrm{in}} - p_{\mathrm{ex}}) & \text{if } g^{\star}(x)=g^{\star}(y),\\
        \frac{(p_{\mathrm{ex}} - p_{\mathrm{in}})}{c} & \mathrm{otherwise}.
	\end{cases}
\end{align}
which is positive for nodes belonging to the same block and negative across blocks. 
Accordingly, we maximise the graphon-modularity if we choose the community structure $g(x)$ equal to the planted one $g^{\star}(x)$ and obtain a maximum modularity of

\begin{align}
 Q_{\text{max}} &= \frac{1}{\mu} \frac{c-1}{c^2}(p_{\mathrm{in}} - p_{\mathrm{ex}} ) = \frac{c-1}{c} \frac{p_{\mathrm{in}} - p_{\mathrm{ex}}}{ p_{\mathrm{in}} + (c-1)p_{\mathrm{ex}} }\,,
\end{align}
Thus for all $c>1$ and $p_{\mathrm{ex}}<p_{\mathrm{in}}$, the graphon-modularity is positive as desired for a model with community structure. In Fif.~\ref{fig:SBMgraphonK3}, we show an example PP graphon, for which we indeed perfectly recover the planted partition $g^{\star}(x)$.

\subsubsection{Modularity of graphs sampled from planted partition models and inferred graphons}
We now consider the numerical performance of modularity optimization for a PP model with $c=3$ groups in three scenarios:
(i) if we had access to the correct graphon, (ii) for a graph sampled from such a graphon, and (iii) for a graphon inferred from such a sampled graphon.

We start by considering the simple baseline case, in which we have access to an appropriately discretized version of the true graphon.
In this case, we can simply use the {\sc GenLouvain} algorithm, or any other approach and recover the $c=3$ planted communities, as long as $p_{\mathrm{ex}}<p_{\mathrm{in}}$ (results not shown).

\begin{figure}[t]
{\centering \includegraphics[width=0.99\textwidth]{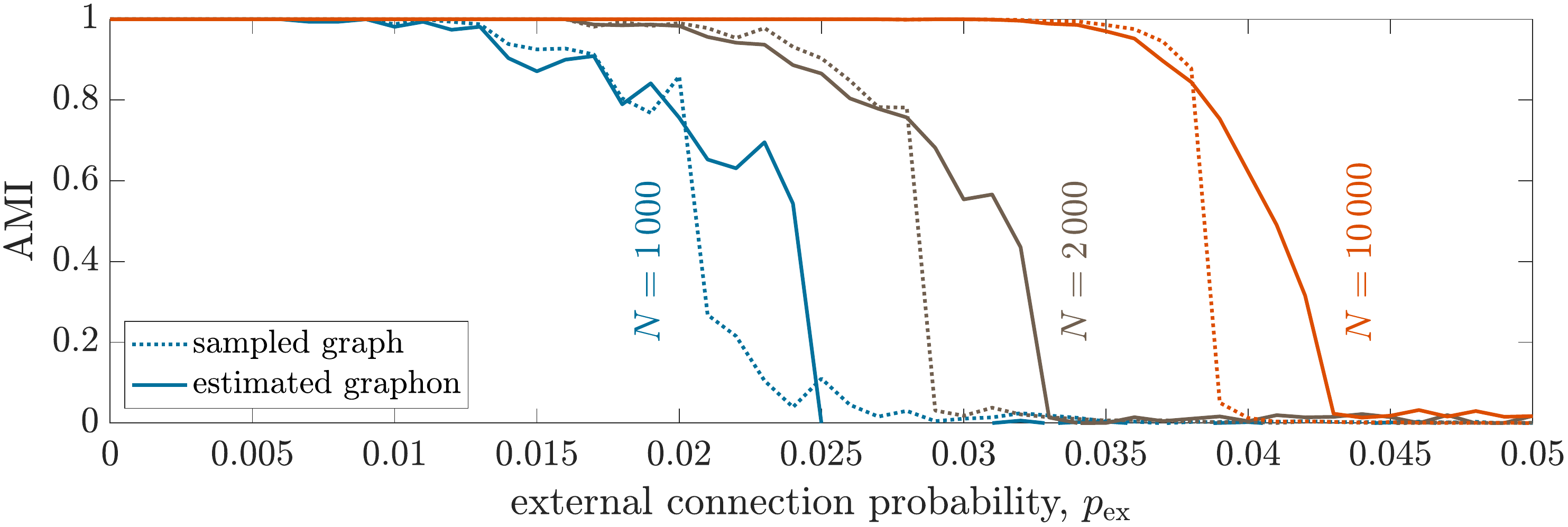}
    \caption{\textbf{Numerical comparison of graphon modularity and graph-based modularity optimization.} We measure the alignment between community structure in graphs sampled from the graphon with the planted partition (dotted lines), and graphons estimated from the sampled graphs (solid lines) as the adjusted mutual information (AMI). 
        The larger the size $N$ of the sampled networks, the better is the recovery of the planted partition. In all cases, there exists a regime for which the graphon estimation improves the detected community structure. Here we fix $p_\text{in}=0.05$.}
  \label{fig:SBMgraphonPvary}
  }
\end{figure}

Next we consider modularity optimization for graphs sampled from such such a planted partition model graphon, and for the graphons inferred from such samples.
We construct planted partition graphons with $c=3$ planted partitions and internal connection probability $p_{\mathrm{in}}=0.05$, with varying external link probability $p_{\mathrm{ex}} \in [0,0.05]$. 
Then we sample networks of varying sizes $N \in \{1000, 2000, 10000\}$ from this graphon as follows.
First, we create $N$ nodes with associated uniformly spaced coordinates in $[0,1]$, such that $x_i= (i-1)/(N-1)$  is the position of node $i$.  
We then draw edges between all (unordered) pairs $(x_i,x_j)$ of nodes with probability $W(x_i,x_j)$ to obtain the symmetric adjacency matrix of an undirected graph.
Finally, we estimate the graphon from the sampled graph with a matrix-completion approach, which we choose because it is fast and does not have free hyperparameters~\cite{keshavan2010matrix}. 

We use {\sc GenLouvain} to optimise the modularity function of the sampled graphs and the inferred graphons and compare the detected partition with the planted one by computing the \emph{adjusted mutual information} (AMI)~\cite{vinh2010information}. 
A maximal value of AMI$=1$ indicates that two partitions are identical and a minimal value of AMI$=0$ indicates that two partitions do not provide more information about each other than expected by random chance.

The results of our numerical comparison are shown in Fig.~\ref{fig:SBMgraphonPvary}.
The dotted lines correspond to the results of modularity optimization on the sampled graphs, the solid lines show the results when we first infer a graphon and second apply modularity optimization on the inferred graphon.

We find that for modularity optimization both on the sampled graph as well as the inferred graphon, for small external connection probabilities $p_{\mathrm{ex}}$ the AMI is one and thus we recover the planted partitions perfectly. 
For larger connection probabilities $p_{\mathrm{ex}}$ the AMI decreases, which indicates that we do not fully recover the planted partitions.
We further see that if we increase the number of sampled nodes, the results obtained from modularity optimization improve, and we are able to find the correct partition even for larger values of $p_\text{ex}$ (i.e., for smaller differences $p_\text{in} - p_\text{ex}$), corresponding to the fact that the sampled graphs converge to the graphon, for which the detection of the planted structures is always possible if $p_\text{in}> p_\text{ex}$.
Indeed, it is known that for dense graphs described by graphons the accurate detection of planted communities is a problem that can be efficiently solved by many algorithms, in contrast to the case of sparse graphs for which a detectability limit exists~\cite{decelle2011inference}.

We find that the AMI curves for the estimated graphons follow a behaviour similar to the AMI curves for modularity optimization on the the sampled graphs.
Interestingly, however, for all considered graph sizes $N$ there exists a range of connection probabilities $p_{\mathrm{ex}}$ for which modularity optimization on the estimated graphon yields a better performance.
This indicates that graphon estimation can improve the performance of modularity optimization for recovering planted partitions in graphs by smoothing fluctuations in the observed connectivity structure.

\section{Modularity optimization and community structure for graphons estimated from empirical data}\label{subseq:empiricalExamples}
Real-world network data can be large but is always of finite size. 
We thus cannot observe graphons directly but rather finite graphs consisting of discrete nodes and edges.
Accordingly, we need to estimate graphons from finite observations.
Many different methods have been proposed for this estimation task~\cite{wolfe2013nonparametric,olhede2014network,zhang2017estimating,chan2014consistent}. 
Most of these graphon estimators are consistent, i.e., the estimation error vanishes as the number of nodes $N \rightarrow \infty$. 
However, different graphon estimators may estimate different graphons for the same finite graph, and identifying the most appropriate estimator for a certain data set is an open research question~\cite{gao2015rate}.

In the following, we employ three prototypical methods for graphon estimation from empirically observed graphs: (i) a \emph{sorting-and-smoothing} algorithm, which is a consistent histogram estimator~\cite{chan2014consistent}, (ii) a matrix-completion approach~\cite{keshavan2010matrix}, and (iii) \emph{universal singular value thresholding} (USVT)~\cite{chatterjee2015matrix}. We choose these three methods because they do not have hyperparameters that have to be chosen by the user.
Using the estimated graphons we then employ modularity maximisation using the {\sc GenLouvain} approach to obtain a simplified picture of the graphons in terms of community structure.
We find that the graphon-estimation approach can have a strong influence on the community structure detected by graphon-modularity maximisation. 

To quantify the extent to which the community structure $g_{\mathrm{graphon}}(x)$ obtained from the estimated graphon resembles the community structure $g_{\mathrm{network}}(x)$ obtained from a graph itself, we compute the AMI between both for six empirical networks (see Table~\ref{tab:normalisedInformation}). 
It is important to note that there is no ground truth in this setting~\cite{peel2017ground}, i.e., the communities obtained from direct modularity optimization on the observed graph are but one possible clustering of the network, similar to the clustering obtained from the graphon.
In fact in some cases one may even argue that the graphon modularity results are less prone to random fluctuations as the estimation strategies involved typically involve some kind of smoothing procedure (as observed for the PP graphon in \cref{subsec:pp}).
We find, that for all networks, the choice of the graphon estimator has a strong influence on the community structure that we detect in the graphon. 
The sort-and-smooth estimator leads to a community structure that differs strongly from the one detected in the network itself, as indicated by small AMI values. 
Modularity maximization using an inferred graphon based on a matrix-completion approach yields partitions that are commensurate with the partitions found from direct modularity optimization on the graph ($\text{AMI}>0.4$) for all data sets.
The ``best results'' is this sense are provided by the USVT estimator.
Interestingly, for the data sets analysed  here, the USVT estimator always yields the highest AMI if it does not give a zero result.
This zero AMI score occurs if the USVT estimator returns a constant graphon $W(x,y)=\text{const.}$, which means all partitions will have the same modularity score of $0$ (and thus no partition will be detected).
This suggests that the USVT algorithm is a good `first choice' for the estimation step of graphons for modularity optimization.
In the following we concentrate on the matrix completion and the USVT estimators as the sort-and-smooth approach yields results that are largely incomparable to the other methods.

Our results indicate that if there is strong community structure, direct modularity optimization and graphon-based modularity optimization yield the same results.
For instance, in the US-senate voting network~\cite{waugh2009party}, which has a strong community structure, we find that the modularity maximisation for graphon and network yield virtually the same partition.
In cases when the community structure in the graphon differs from the one in the network, we usually detect less communities in the graphon because the graphon-estimation smoothes some of the connectivity signal.
For the Facebook network the AMI is $0.45$, which means that more than half of the information of the graph partition can be revealed by clustering the graphon. 
For the brain connectivity network~\cite{hagmann2008mapping}, which has shown to have a modular structure that enables the parallel processing of the information~\cite{klimm2014individual,klimm2014resolving}, we find intermediate AMI of $0.54$. 
Our analysis indicates, that a privacy-preserving community detection via graphon modularity is indeed possible, but extent to which we loose information in comparison to the graph partition depends on the data set and the type of graphon-estimation algorithm. 
We postpone a more detailed investigation of this behaviour for future work.

\begin{table}[t]
\centering
\begin{small} 
\begin{tabular}{c| c | c || c  c  c} 
data set & reference & $N$ & \multicolumn{3}{c}{AMI($g_{\mathrm{network}}$, $g_{\mathrm{graphon}}$)}\\
 & & & Matrix completion & USVT & Sort-and-smooth \\ \hline
Zachary Karate Club & \cite{zachary1977information} & 34 & ${\mathbf{0.43}}$ &\~0 & 0.09 \\
Senate voting & \cite{waugh2009party} & 102 & 0.9286 & ${\mathbf{1}}$ & 0.01 \\
Facebook network & \cite{maier2017cover} & 329 & ${\mathbf{0.45}}$ & 0 & 0 \\
Brain connectivity & \cite{hagmann2008mapping} & 998 & 0.48 & ${\mathbf{0.54}}$ & 0.14 \\
Political blog & \cite{adamic2005political} & 1224 & 0.60 & ${\mathbf{0.79}}$ & 0.17 \\
Protein complex & \cite{klimm2020hypergraphs} & 8243 & 0.44 & ${\mathbf{0.60}}$ & 0.04\\
\end{tabular}
\vspace{0.2cm}
\end{small}
\caption{The agreement between partitions obtained from clustering a graph directly and clustering a graphon estimated from the graph depends strongly on the data set and the used estimation algorithm. We show the adjusted mutual information (AMI) between the partition $g_{\mathrm{network}}$ obtained from a graph and the partition obtained from the graphons $g_{\mathrm{graphon}}$ that we estimated from the graph, using three different approaches. We also provide a reference for the data and the number $N$ of nodes in each network. For each data, we highlight the estimation method that yields the best result. While the matrix-completion method always yields decent results ($\text{AMI}>0.4$), the USVT is better for some data sets. For all data sets, the sort-and-smooth algorithm yields small AMIs.}
\label{tab:normalisedInformation}
\end{table}

\section{Discussion}\label{sec:discussion}
In this manuscript, we considered the problem of modularity optimization from the perspective of graphons. 
We showed how a generalised modularity-maximisation algorithm for graphs can be used for modularity optimization on graphons after suitable discretization and discussed how in certain cases analytical solutions for the modularity optimization problem for graphs can be obtained.
For future research, exploring how far further insights into trace maximisation problems, such as modularity optimization, can be obtained by using such a perspective based on operators defined on a continuous domain would be of interest.
Interestingly, it has been shown recently that maximum-likelihood estimation of an stochastic block model (a problem closely related to Modularity optimization~\cite{pamfil2019relating,newman2016equivalence}) is equivalent to a discrete surface tension~\cite{boyd2019stochastic}, thus providing a connection to partial differential equations and continuous problem formulations.

There are also other avenues to explore in the future:
In this manuscript, we discussed graphon-modularity with the popular Newman--Girvan null model. 
Our framework does, however, also allow the use of other null models or a resolution parameter, which might reveal a hierarchical community structure in graphons. 

One limitation of graphons is that they describe limits of dense graphs~\cite{orbanz2014bayesian}.
Many empirical networks, however, are sparse. 
\emph{Exchangeable random measures} have been proposed as a way to construct sparse graph-variants~\cite{caron2017sparse}. 
It would be relevant for applications, to extend modularity-based approaches to also detect community structure in these objects.

From a privacy-preserving computing point of view, it would be interesting to explore to what extent a graphon description could be de-anonymised, when obtaining information about the network from which it was estimated.
Furthermore, an investigation in how far other graph-measures that can be extended to graphons allow for privacy-presering computation would be interesting.


\section{Code availability}

{\sc Matlab} and {\sc Mathematica} code to implement the discussed methods and reproduce all figures is available under \url{http://github.com/floklimm/graphon}. 

\bibliography{graphonCommunity} 
\bibliographystyle{unsrt} 

\clearpage
\appendix

\section{Modularity slices for noncontiguously ordered communities}

In the main text, we introduce the sliced-modularity approach for contiguously ordered communities. Here, we demonstrate that it is possible to define modularity slices for noncontiguously ordered communities. This might be necessary since, while we can hope that methods for inferring graphons will notionally assign close coordinates to subsets of nodes from the same community, the algorithms for inferring graphons will not always assign all nodes in a community a set of contiguous co-ordinates; rather, we might have sets of nodes from the same community to be cut and so be fragmented into contiguous subsets and these subsets assigned noncontiguous locations in the graphon embedding (see~Fig.~\ref{fig:SI} for an example).

\begin{figure}[b!]
{\centering \includegraphics[width=0.49\textwidth]{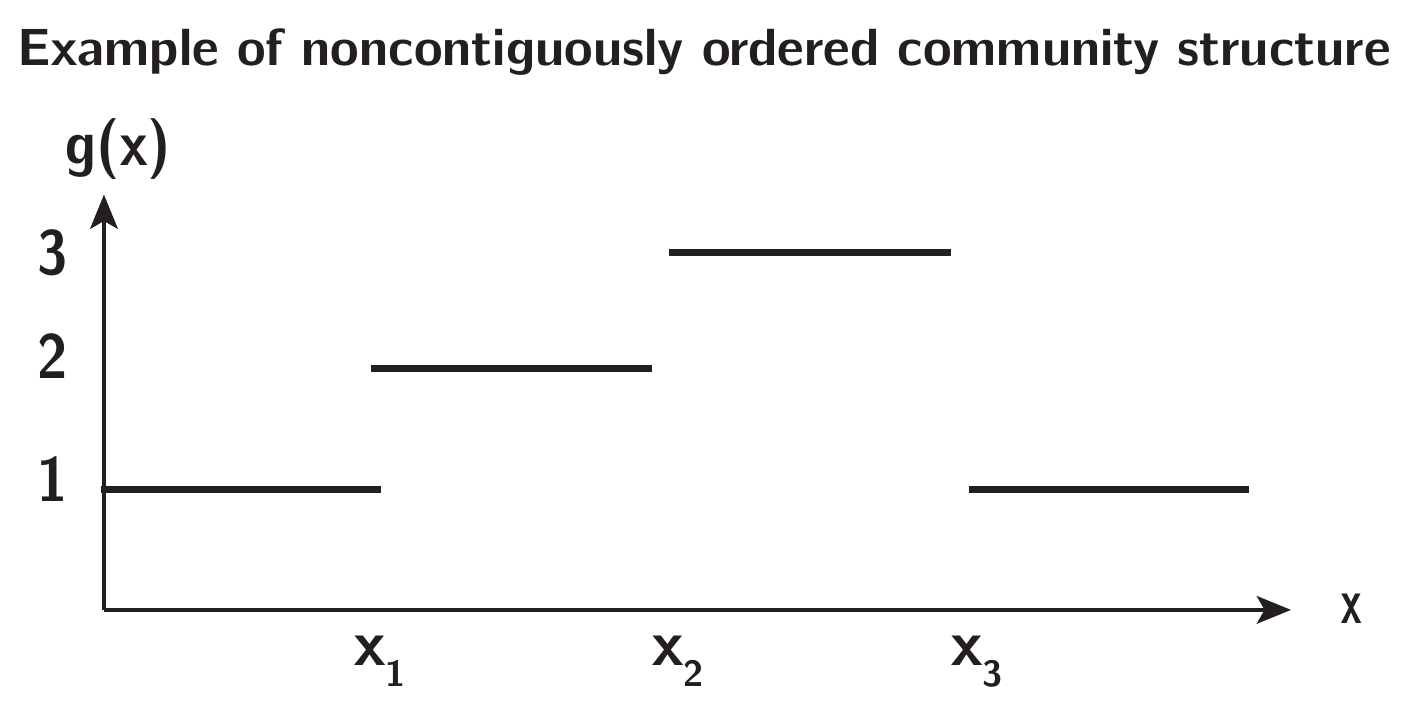}
    \caption{\textbf{The sliced-modularity approach can be extended to encompass noncontiguously ordered communities.} In this example, there exist $s=4$ slices, which are sorted into $c=3$ communities such that $G(1)=G(4)=1$, $G(2)=2$, and $G(3)=3$. Using the sliced-modularity approach for discontinuous communities, we now can optimise the boundaries $x_1$, $x_2$, and $x_3$ between the communities.}
  \label{fig:SI}
  }
\end{figure}

To this end, we restrict the possible group assignment functions $g(x)$ to be piecewise constant on $s$ intervals:
\begin{equation}
    g(x) = 
    \begin{cases}
        G(1) \quad \text{if} \quad x\in[0,x_1)\\
        G(2) \quad \text{if} \quad x\in[0,x_1)\\
        \vdots \\
        G(s) \quad \text{if} \quad x\in[x_{s-1},1]\,,\\
    \end{cases}
\end{equation}
where the, possibly many-to-one, \emph{slice--community} function $G: \{1,2,\dots,s\} \to \{1,2,\dots,c\}$ maps each slice to one of the $c$ communities with $c\leq s$. This allows noncontiguously ordered communities, for example with $s=4$ slices but $c=3$ communities, in which the first slice and the fourth slice are in the same community, such that $G(1)=G(4)\neq G(2) \neq G(3)$ (see Fig.~\ref{fig:SI}). For $c=s$, each slice belongs to its own community and we obtain the case of continuous communities, as discussed in the main text.

If the group assignment function $g(x)$ is of such a form, we can rewrite the graphon modularity function as:
\begin{align}
	Q(g) = \frac{1}{2\mu} &\sum_{i=1}^s \sum_{j=1}^s L^{(2)}(x_{i-1},x_i,x_{j-1},x_j)\delta[ G(i),G(j) ]  \,,\ \text{with} \\
	L^{(2)}(a,b,c,d) =& \int_{a}^{b} \int_{c}^{d} B(x,y) dx dy\,,
\end{align}
where the modularity slice $L^{(2)}(a,b,c,d)$ is now a function of four interval points.

If we fix the slice--community function $G(s')$ \emph{apriori} (i.e.,~we know the order of intervals that belong to the same community), we have to optimise the community borders $x_i$, which is an optimization with $s-1$ degrees of freedom, very similar to the contiguous case in the main manuscript. A more challenging setting occurs when $G(s')$ is unknown (i.e., we have to optimise over all possible slice--community functions $G(s')$, as well as, the community borders $x_i$). For a small number $s$ of slices, we may test all possible partitions but for a large number $s$ of slices this might become unfeasible and heuristics would be required. This is especially a problem, if the appropriate number $s$ of slices is unknown. This points to an algorithm that is initialised by assuming that all slices are associated with distinct communities and assumes that $s$ is large (larger than the true number of communities): the algorithm then finds the boundaries between the slices and might then attempt recursive mergers of the slices to maximize the modularity.
\end{document}